\documentclass[letterpaper,11pt]{article}
\usepackage{tabularx} 
\usepackage{amsmath}  
\usepackage{amsthm}
\usepackage{amsfonts}
\usepackage{tikz}
\usetikzlibrary{fit,graphs,graphs.standard,quotes,shapes,arrows,automata}
\usepackage{graphicx} 
\usepackage[margin=1.3in,letterpaper]{geometry} 
\usepackage{cite} 
\usepackage[final]{hyperref} 
\hypersetup{
	colorlinks=true,       
	linkcolor=red,        
	citecolor=blue,        
	filecolor=magenta,     
	urlcolor=blue         
}
\newtheorem{theorem}{Theorem}
\newtheorem{lemma}{Lemma}

\newtheorem{corollary}{Corollary}
\newtheorem{proposition}{Proposition}

\definecolor{bblue}{rgb}{0.2, 0.4, 0.8}
\definecolor{bgreen}{rgb}{0.2, 0.8, 0.2}
\definecolor{bred}{rgb}{0.8, 0.2, 0.2}
\definecolor{lgreen}{rgb}{0.0, 0.48, 0.0}
\definecolor{lpurple}{rgb}{0.48, 0.0, 0.48}

\tikzset{
	treenode/.style = {align=center, inner sep=0pt, text centered,
		font=\sffamily},
	arn_nn/.style = {treenode, circle, bblue, draw=bblue, 
		fill=bblue!10,
		minimum width=0.5em, minimum height=0.5em
	},
	arn_n/.style = {treenode, circle, bblue, draw=bblue, 
		text width=2.0em, very thick, 
		fill=bblue!10},
	arn_g/.style = {treenode, circle, bgreen, draw=bgreen, 
		text width=1.5em, very thick,
		fill=bblue!10},
	arn_r/.style = {treenode, circle, bred, draw=bred, 
		text width=2.0em, very thick,
		fill=bred!10},
	arn_x/.style = {treenode, rectangle, draw=black,
		minimum width=0.5em, minimum height=0.5em}
}

\begin{document}
	
\title{Synchronization Problems in Automata\\ without Non-trivial Cycles}
\author{Andrew Ryzhikov$^{1, 2}$}
\date{%
	$^1$Universit\'e Grenoble Alpes, Laboratoire G-SCOP, 38031 Grenoble, France\\%
	$^2$United  Institute of Informatics Problems  of NASB, 
	220012 Minsk, Belarus\\
	{\tt ryzhikov.andrew@gmail.com}
}
\maketitle

\begin{abstract}
	We study the computational complexity of various problems related to synchronization of weakly acyclic automata, a subclass of widely studied aperiodic automata. We provide upper and lower bounds on the length of a shortest word synchronizing a weakly acyclic automaton or, more generally, a subset of its states, and show that the problem of approximating this length is hard. We investigate the complexity of finding a synchronizing set of states of maximum size. We also show inapproximability of the problem of computing the rank of a subset of states in a binary weakly acyclic automaton and prove that several problems related to recognizing a synchronizing subset of states in such automata are NP-complete.
\end{abstract}

\section{Introduction}

The concept of synchronization is widely studied in automata theory and has a lot of different applications in such areas as manufacturing, coding theory, biocomputing, semigroup theory and many others \cite{Volkov2008}. Let $A = (Q, \Sigma, \delta)$ be a complete deterministic finite automaton (which we simply call an {\em automaton} in this paper), where $Q$ is a set of states, $\Sigma$ is a finite alphabet and $\delta: Q \times \Sigma \to Q$ is a transition function. Note that our definition of an automaton does not include initial and accepting states. The function $\delta$ can be naturally extended to a mapping $Q \times \Sigma^* \to Q$, which we also denote as $\delta$, in the following way: for $x \in \Sigma$ and $a \in \Sigma^*$ we recursively set $\delta(q, xa) = \delta(\delta(q, x), a)$. An automaton is called {\it synchronizing} if there exists a word that maps all its states to a fixed state. Such word is called a {\em synchronizing word}. A state $q \in Q$ is called a {\em sink state} if all letters from~$\Sigma$~map~$q$ to itself.

In this paper synchronization of weakly acyclic automata is studied. A {\em simple cycle} in an automaton $A = (Q, \Sigma, \delta)$ is a sequence $q_1, \ldots, q_k$ of its states such that all the states in the sequence are different and there exist letters $x_1, \ldots, x_k \in \Sigma$ such that $\delta(q_i, x_i) = q_{i + 1}$ for $1 \le i \le k - 1$ and $\delta(q_k, x_k) = q_1$. A simple cycle is a {\em self-loop} if it consists of only one state. An automaton is called {\em weakly acyclic} if all its simple cycles are self-loops. In other words, an automaton is weakly acyclic if and only if there exists an ordering $q_1, q_2, \ldots, q_n$ of its states such that if $\delta(q_i, x) = q_j$ for some letter $x \in \Sigma$, then $i \le j$ (such ordering is called a {\em topological sort}). Since a topological sort can be found in polynomial time \cite{Cormen2009}, this class can be recognized in polynomial time. Weakly acyclic automata are called acyclic in \cite{Jiraskova2012} and partially ordered in \cite{Brzozowski1980}, where in particular the class of languages recognized by such automata is characterized.

Weakly acyclic automata arise naturally in synchronizing automata theory. Section \ref{sec-short-words-compl} of this paper shows several examples of existing proofs where weakly acyclic automata appear implicitly in complexity reductions. Surprisingly, most of the computational problems that are hard for general automata remain very hard in this class despite its very simple structure. Thus, investigation of weakly acyclic automata provides good lower bounds on the complexity of many problems for general automata.  An automaton is called {\em aperiodic} if for any word $w \in \Sigma^*$ and any state $q \in Q$ there exists $k$ such that $\delta(q, w^k) = \delta(q, w^{k + 1})$, where $w^k$ is a word obtained by $k$ concatenations of $w$ \cite{Trahtman2007}. Obviously, weakly acyclic automata form a proper subclass of aperiodic automata, thus all hardness results hold for the class of aperiodic automata.

The concept of synchronization is often used as an abstraction of returning control over an automaton when there is no a priori information about its current state, but the structure of the automaton is known. If the automaton is synchronizing, we can apply a synchronizing word to it, and thus it will transit to a known state. If we want to perform the same operation when the current state is known to belong to some subset of states of the automaton, we come to the definition of a synchronizing set. A set $S \subseteq Q$ of states of an automaton $A$ is called {\em synchronizing} if there exists a word $w \in \Sigma^*$ and a state $q \in Q$ such that the word $w$ maps each state $s \in S$ to the state $q$. The word $w$ is said to {\em synchronize} the set $S$. It follows from the definition that an automaton is synchronizing if and only if the set $Q$ of all its states is synchronizing. Consider the problem {\sc Sync Set} of deciding whether a given set $S$ of states of an automaton $A$ is synchronizing.

\begin{tabular}{||p{30em}}
	~{\sc Sync Set}\\
	~{\em Input}: An automaton $A$ and a subset $S$ of its states;\\
	~{\em Output}: Yes if $S$ is a synchronizing set, No otherwise.
\end{tabular}

The {\sc Sync Set} problem is PSPACE-complete \cite{Rystsov1983, Sandberg2005}, even for binary strongly connected automata \cite{Vorel2016} (an automaton is called {\em binary} if its alphabet has size two, and {\em strongly connected} if any state can be mapped to any other state by some word). In \cite{Natarajan1986} it is shown that the {\sc Sync Set} problem is solvable in polynomial time for orientable automata if the cyclic order respected by the automaton is provided in the input. This problem is also solvable in polynomial time for monotonic automata \cite{Ryzhikov2017Monotonic}. The problem of deciding whether the whole set of states of an automaton is synchronizing is also solvable in polynomial time \cite{Volkov2008}.

One of the most important questions in synchronizing automata theory is the famous \v{C}ern{\'y} conjecture stating that any $n$-state synchronizing automaton has a synchronizing word of length at most $(n - 1)^2$. The conjecture is proved for various special cases, including orientable, Eulerian, aperiodic and other automata (see \cite{Volkov2008} for references), but is still open in general. For more than 30 years, the best upper bound was $\frac{n^3 - n}{6}$, obtained in \cite{Pin1983}. Recently, a small improvement on this bound has been reported in~\cite{Szykula2017}: the new bound is still cubic in $n$ but improves the coefficient $\frac{1}{6}$ at $n^3$ by $\frac4{46875}$.

While there is a simple cubic bound on the length of a synchronizing word for the whole automaton, there exist examples of automata where the length of a shortest word synchronizing a subset of states is exponential in the number of states \cite{Vorel2016}. For orientable $n$-state automata, a tight upper bound of $(n - 1)^2$ is known \cite{Eppstein1990}, and this bound is also asymptotically tight for monotonic automata \cite{Ryzhikov2017Monotonic}. On the other hand, a trivial upper bound $2^n - n - 1$ on the length of a shortest word synchronizing a subset of states in a $n$-state automaton is known \cite{Vorel2016}. In \cite{CardosoThesis2014} Cardoso considers the length of a shortest word synchronizing a subset of states in a synchronizing automaton.

We assume that the reader is familiar with the notions of an NP-complete problem (refer to the book by Sipser~\cite{Sipser2006}), an approximation algorithm and a gap-preserving reduction (for reference, see the book by Vazirani~\cite{Vazirani2001}). 

Given an automaton $A$, the {\em rank} of a word $w$ with respect to $A$ is the number $|\{\delta(s, w) \mid s \in Q\}|$, i.e., the size of the image of $Q$ under the mapping defined in $A$ by $w$. More generally, the rank of a word $w$ with respect to a subset $S$ of states of $A$ is the number $|\{\delta(s, w) \mid s \in S\}|$. The {\em rank} of an automaton (resp. of a subset of states) is the minimum among the ranks of all words $w \in \Sigma^*$ with respect to the automaton (resp. to the subset of states).

In this paper we provide various results concerning computational complexity and approximability of the problems related to subset synchronization in weakly acyclic automata. In Section \ref{short-words} we prove some lower and upper bounds on the length of a shortest word synchronizing a weakly acyclic automaton or, more generally, a subset of its states. In Section \ref{sec-short-words-compl} we investigate the computational complexity of finding such words. In Section \ref{sec-max-sync-set} we study inapproximability of the problem of finding a subset of states of maximum size. In Section \ref{sec-rank} we give strong inapproximability results for computing the rank of a subset of states in binary weakly acyclic automata. In Section \ref{sec-subset} we show that several other problems related to recognizing a synchronizing set in a weakly acyclic automaton are hard.

A preliminary conference version of this paper was published in \cite{Ryzhikov2017WA}.

\section{Bounds on the Length of Shortest Synchronizing Words}\label{short-words}

Each synchronizing weakly acyclic automaton is a $0$-automaton (i.e., an automaton with exactly one sink state), which gives an upper bound $\frac{n(n - 1)}{2}$ on the length of a shortest synchronizing word \cite{Rystsov1997}. The same bound can be deduced from the fact that each weakly acyclic automaton is aperiodic \cite{Trahtman2007}. However, for weakly acyclic automata a more accurate result can be obtained, showing that weakly acyclic automata of rank $r$ behave in a way similar to monotonic automata of rank $r$ (see \cite{Ananichev2004}).

\begin{proposition}\label{lemma-mon-log-word}
	Let $A = (Q, \Sigma, \delta)$ be a $n$-state weakly acyclic automaton, such that there exists a word of rank~$r$ with respect to $A$. Then there exists a word of length at most $n - r$ and rank at most $r$ with respect~to~$A$.
\end{proposition}
\begin{proof} Observe that the rank of a weakly acyclic automaton is equal to the number of sink states in it. The conditions of the theorem imply that $A$ has at most $r$ sink states.
	
	 Consider the sets $S_1, \ldots, S_t$ constructed in the following way. Let $p_i$ be the state in $S_{i - 1}$ with the smallest index in the topological sort such that $p_i$ is not a sink state. Let $x_i, 1 \le i \le t$, be a letter mapping the state $p_i$ to some other state, where $S_i = \{\delta(q, x_i) \mid q \in S_{i - 1}\}, 1 \le i \le t$, and $S_0 = Q$. Since $A$ has at most $r$ sink states, the word $w = x_1 \ldots x_t$ exists for any $t \le n - r$ and has rank at most $r$ with respect to $A$. 
\end{proof}

The following simple example shows that the bound is tight. Consider an automaton $A = (Q, \Sigma, \delta)$ with states $q_1, \ldots, q_n$. Let each letter except some letter $x$ map each state to itself. For the letter $x$ define the transition function $\delta(q_i, x) = q_{i + 1}$ for $1 \le i \le n - r$ and $\delta(q_i, x) = q_i$ for $ n - r + 1 \le i \le n$. Obviously, $A$ has rank $r$ and shortest words of rank $r$ with respect to $A$ have length $n - r$.

\begin{proposition}\label{thm-ac-subs-short-word}
	Let $S$ be a synchronizing set of states of size $k$ in a weakly acyclic $n$-state automaton~$A = (Q, \Sigma, \delta)$. Then the length of a shortest word synchronizing $S$ is at most $\frac{k(2n - k - 1)}{2}$.
\end{proposition}
\begin{proof} Consider a topological sort $q_1, \ldots, q_n$ of the set $Q$. Let $q_s$ be a state such that all states in $S$ can be mapped to it by a shortest word $w = x_1 \ldots x_t$. We can assume that the images of all words $x_1 \ldots x_j$, $j \le t$, are pairwise distinct, otherwise some letter in this word can be removed. Then a letter $x_j$ maps at least one state of the set $\{\delta(q, x_1 \ldots x_{j - 1}) \mid q \in S \}$ to some other state. Thus the maximum total number of letters in $w$ sending all states in $S$ to $q_s$ is at most $(n - k) + (n - k + 1) + \ldots + (n - 1) = \frac{k(2n - k - 1)}{2}$, since application of each letter of $w$ increases the sum of the indices of reached states by at least one.  \end{proof}

Consider a binary automaton $A = (Q, \{0, 1\}, \delta)$ with $n$ states $q_1, \ldots, q_{k - 1}$, $s_1, \ldots, s_{\ell}$, $t$, where $\ell = n - k$. Define $\delta(q_i, 0) = q_i$, $\delta(q_i, 1) = q_{i + 1}$ for $1 \le i \le k - 2$,  $\delta(q_{k - 1}, 1) = s_1$. Define also $\delta(s_i, 0) = s_{i + 1}$ for $1 \le i \le \ell - 1$, $\delta(s_i, 1) = t$ for $1 \le i \le \ell - 1$. Define both transitions for $s_\ell$ and $t$ as self-loops. Set $S = \{q_1, \ldots, q_{k - 1}, s_\ell\}$. The shortest word synchronizing $S$ is $(10^{l - 1})^{k-1}$ of length $(k-1)(n - k)$. The automaton in this example is binary weakly acyclic, and even has rank $2$. Figure \ref{fig-subset} gives the idea of the described construction.

\setlength{\unitlength}{2pt}
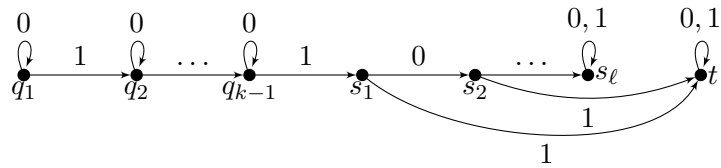
\begin{figure}[hbt]
	\begin{center}
		\begin{tikzpicture}[->,>=latex',
		vertex/.style={circle, draw=black, fill, minimum width=1.5mm, inner sep=0pt, outer sep=0pt},
		every label/.style={inner sep=0pt, minimum width=0pt, label distance=0.1mm},
		yscale=1,
		xscale=1.5
		]
		\graph[nodes=vertex, empty nodes, no placement] {

			{
				q1[x=0,y=0,label=below:$q_1$] -> [edge label=1]
				q2[x=1,y=0,label=below:$q_2$] -> [edge label=$\ldots$]
				qk[x=2,y=0,label=below:$q_{k-1}$] -> [edge label=1]
				s1[x=3,y=0,label=below:$s_1$] -> [edge label=0]
				s2[x=4,y=0,label=below:$s_2$] -> [edge label=$\ldots$]
				sl[x=5,y=0,label=right:$s_\ell$]
			};
			{
				t[x=6,y=0,label=right:$t$] -> [edge label={$0, 1$}, loop above] t
			};
			{
				q1 -> [edge label=0, loop above] q1
			};
			{
				q2 -> [edge label=0, loop above] q2
			};
			{
				qk -> [edge label=0, loop above] qk
			};
			{
				s1 -> [edge label=1, swap, bend right, out=310, in=250] t
			};
			{
				s2 -> [edge label=1, swap, bend right] t
			};
			{
				sl -> [edge label={$0, 1$}, loop above] s3
			};
		};
		\end{tikzpicture}
		\caption{The automaton providing the lower bound for subset synchronization} \label{fig-subset}
	\end{center}
\end{figure}

As was noted by an anonymous reviewer, for alphabet of size $n-2$, a better lower bound of $\frac{(k - 1)(2n - k - 2)}{2}$ can be shown as follows. Let $Q = \{-1, 0, 1, \dots, n-2\}$, $\Sigma = \{a_1, \dots, a_{n - 2}\}$,

$$\delta(k, a_i) = \left\{
	\begin{array}{ll}
	k &\mbox{if } k > i,\\
	k - 1 &\mbox{if } k = i,\\
	-1  &\mbox{if } 0 < k < i,\\
	k &\mbox{if } k\in\{-1, 0\}
	\end{array}
\right.$$

If $k < n$ and $S=\{0, n - 2, n - 3, \dots, n - k\}$, then it is easy to see that the shortest word synchronizing $S$ has length $(n - k) + (n - k + 1) + \dots + (n - 2) = \frac{(k - 1)(2n - k - 2)}{2}$. For each $n$ and $k$, this is less than the upper bound of Proposition \ref{thm-ac-subs-short-word} by $n - 1$ only.

\section{Complexity of Finding Shortest Synchronizing Words}\label{sec-short-words-compl}

Now we proceed to the computational complexity of some problems, related to finding a shortest synchronizing word for an automaton. Consider first the following problem.

\begin{tabular}{||p{30em}}
	~{\sc Shortest Sync Word}\\
	~{\em Input}: A synchronizing automaton $A$;\\
	~{\em Output}: The length of a shortest synchronizing word for $A$.
\end{tabular}

First, we note that the automaton showing inapproximability of {\sc Shortest Sync Word} in the construction of Berlinkov \cite{Berlinkov2014} is weakly acyclic.

\begin{proposition}
	For any $\gamma > 0$, the {\sc Shortest Sync Word} problem for $n$-state weakly acyclic automata with alphabet of size at most $n^{1 + \gamma}$ cannot be approximated in polynomial time within a factor of $d \log n$ for any $d < c_{sc}$ unless P = NP, where $c_{sc}$ is some constant. 
\end{proposition}

In Berlinkov's reduction to the binary case, the automaton is no longer weakly acyclic. However, the binary automaton showing NP-hardness of {\sc Shortest Sync Word} in Eppstein's construction \cite{Eppstein1990} is weakly acyclic.

\begin{proposition}
	{\sc Shortest Sync Word} is NP-hard for binary weakly acyclic automata.
\end{proposition}

Consider now the following more general problem.

\begin{tabular}{||p{30em}}
	~{\sc Shortest Set Sync Word} \\
	~{\em Input}: An automaton $A$ and a synchronizing subset $S$ of its states; \\
	~{\em Output}: The length of a shortest word synchronizing $S$.
\end{tabular}

It follows from Proposition \ref{thm-ac-subs-short-word} that the decision version of this problem (asking whether there exists a word of length at most $k$ synchronizing $S$) is in NP for weakly acyclic automata, so it is reasonable to investigate its approximability.

\begin{theorem}
	The {\sc Shortest Set Sync Word} problem for $n$-state binary weakly acyclic automata cannot be approximated in polynomial time within a factor of $O(n^{\frac{1}{2} - \epsilon})$ for any $\epsilon > 0$ unless P = NP.
\end{theorem}

\begin{proof} To prove this theorem, we construct a gap-preserving reduction from the  {\sc Shortest Sync Word} problem in $p$-state binary automata, which cannot be approximated in polynomial time within a factor of $O(p^{1 - \epsilon})$ for any $\epsilon > 0$ unless P = NP \cite{Gawrychowski2015}. Let a binary automaton $A = (Q, \{0, 1\}, \delta)$ be the input of {\sc Shortest Sync Word}. Let $Q = \{q_1, \ldots, q_p\}$. Construct a binary automaton $A' = (Q', \{0, 1\}, \delta')$ with the set of states $Q' =  \{q_i^{(j)} \mid 1 \le i \le p, 1 \le j \le p + 1\}$. Define $\delta'(q_i^{(j)}, x) = q_k^{(j + 1)}$ for $1 \le i \le p$, $1 \le j \le p$, $x \in \{0, 1\}$, where $k$ is such that $q_k = \delta(q_i, x)$. Define $\delta'(q_i^{(p + 1)}, x) = q_i^{(p + 1)}$ for $1 \le i \le p$ and $x \in \{0, 1\}$. Take $S' = \{q_i^{(1)} \mid 1 \le i \le p \}$.
	
Observe that any word synchronizing $S'$ in $A'$ is a synchronizing word for $A$ because of the definition of $\delta'$. In the other direction, we note that a shortest synchronizing word for a $p$-state automaton in the construction of Gawrychowski and Straszak \cite{Gawrychowski2015} has length at most $p$. Hence, a shortest synchronizing word for $A$ also synchronizes $S'$ in $A'$. Thus, the length of a shortest synchronizing word for $A$ is equal to the length of a shortest word synchronizing $S'$ in $A'$, and we get a gap-preserving reduction with gap $O(p^{1-\epsilon}) = O(n^{\frac{1}{2} - \epsilon})$, as $A'$ has $O(p^2)$ states. Finally, it is easy to see that $A'$ is binary weakly acyclic.  \end{proof}

\section{Finding a Synchronizing Set of Maximum Size} \label{sec-max-sync-set}

One possible approach to measure and reduce initial state uncertainty in an automaton is to find a subset of states of maximum size where the uncertainty can be resolved, i.e., to find a synchronizing set of maximum size. This is captured by the following problem.

\begin{tabular}{||p{30em}}
	~{\sc Max Sync Set} \\
	~{\em Input}: An automaton $A$;\\
	~{\em Output}: A synchronizing set of states of maximum size in $A$.
\end{tabular}

T\"{u}rker and Yenig\"{u}n \cite{Turker2015} study a variation of this problem, which is to find a set of states of maximum size that can be mapped by some word to a subset of a given set of states in a given monotonic automaton. They reduce the {\sc N-Queens Puzzle} problem \cite{Bell2009} to this problem to prove its NP-hardness. However, their proof is unclear, since the input has size $O(\log N)$, and the output size is polynomial in $N$. Also, the {\sc N-Queens Puzzle} problem is solvable in polynomial time \cite{Bell2009}.

First we investigate the PSPACE-completeness of the decision version of the  {\sc Max Sync Set} problem, which we shall denote as {\sc Max Sync Set-D}. Its formulation is the following: given an automaton $A$ and a number $c$, decide whether there is a synchronizing set of states of cardinality at least $c$ in $A$.

\begin{theorem}\label{thm-pspace-general}
	The {\sc Max Sync Set-D} problem is PSPACE-complete for binary automata.
\end{theorem}
\begin{proof} The {\sc Sync Set} problem is in PSPACE \cite{Sandberg2005}. Thus, the {\sc Max Sync Set-D} problem is also in PSPACE, as we can sequentially check whether each subset of states is synchronizing and compare the size of a maximum synchronizing set to $c$.
	
	To prove that the {\sc Max Sync Set-D} problem is PSPACE-hard for binary automata, we shall reduce a PSPACE-complete {\sc Sync Set} problem for binary automata to it \cite{Vorel2014}. Let an automaton $A$ and a subset $S$ of its states be an input to {\sc Sync Set}. Let $n$ be the number of states of $A$. Construct a new automaton $A'$ by initially taking a copy of $A$. For each state $s \in S$, add $n + 1$ {\em new} states to $A'$ and define all the transitions from these new states to map to $s$, regardless of the input letter. Define the set $S'$ to be a union of all new states and take $c = |S'| = (n + 1)|S|$.
	
	Let $S_1$ be a maximum synchronizing set in $A$ not containing at least one new state $q$. As $S_1$ is maximum, it does not contain other $n$ new states that can be mapped to the same state as $q$. Thus, the size of $S_1$ is at most $n + (n + 1)|S| - (n + 1) < (n + 1)|S| = c$. Hence, each synchronizing set of size at least $c$ in $A'$ contains $S'$. The set $S$ is synchronizing in $A$ if and only if $S'$ is synchronizing in $A'$, as each word $w$ synchronizing $S$ in $A$ corresponds to a word $xw$ synchronizing $S'$ in $A'$, where $x$ is an arbitrary letter. Thus, $A'$ has a synchronizing set of size at least $c$ if and only if $S$ is synchronizing in $A$. \end{proof}

Now we proceed to inapproximability results for the {\sc Max Sync Set} problem in several classes of automata. We shall need some results from graph theory. An {\em independent set} $I$ in a graph $G$ is a set of its vertices such that no two vertices in $I$ share an edge. The size of a maximum independent set in $G$ is denoted $\alpha(G)$. The {\sc Independent Set} problem is defined as follows.

\begin{tabular}{||p{30em}}
	~{\sc Independent Set} \\
	~{\em Input}: A graph $G$;\\
	~{\em Output}: An independent set of maximum size in $G$.
\end{tabular}

Zuckerman \cite{Zuckerman2006} has proved that, unless P = NP, there is no polynomial \mbox{$p^{1 - \varepsilon}$-approximation} algorithm for the {\sc Independent Set} problem for any $\varepsilon > 0$, where $p$ is the number of vertices in $G$.

\begin{theorem}\label{thm-inapprox-alph}
	The problem {\sc Max Sync Set} for weakly acyclic $n$-state automata over an alphabet of cardinality $O(n)$ cannot be approximated in polynomial time within a factor of $O(n^{1 - \varepsilon})$ for any~$\varepsilon > 0$ unless P = NP.
\end{theorem}
\begin{proof} We shall prove this theorem by constructing a gap-preserving reduction from the {\sc Independent Set} problem. Given a graph $G = (V, E)$, $V = \{v_1, v_2, \ldots, v_p\}$, we construct an automaton $A = (Q, \Sigma, \delta)$ as follows. For each $v_i \in V$, we construct two states $s_i, t_i$ in $Q$. We also add a state $f$ to $Q$. Thus, $|Q| = 2p + 1$. The alphabet $\Sigma$ consists of letters $\tilde{v}_1, \ldots, \tilde{v}_p$ corresponding to the vertices of $G$.
	
	The transition function $\delta$ is defined in the following way. For each $1 \le i \le p$, the state $s_i$ is mapped to $f$ by the letter $\tilde{v}_i$. For each $v_iv_j \in E$ the state $s_i$ is mapped to $t_i$ by the letter $\tilde{v}_j$, and the state $s_j$ is mapped to $t_j$ by the letter $\tilde{v}_i$. All yet undefined transitions map a state to itself.
	
	\begin{lemma} Let $I$ be a maximum independent set in $G$. Then the set $S = \{s_i \mid v_i \in I\} \cup \{f\}$ is a synchronizing set of maximum cardinality (of size $\alpha(G) + 1$) in the automaton $A = (Q, \Sigma, \delta)$. 
	\end{lemma}
	\begin{proof} Let $w$ be a word obtained by concatenating the letters corresponding to $I$ in arbitrary order. Then $w$ synchronizes the set $S = \{s_i \mid v_i \in I \} \cup \{f\}$ of states of cardinality $|I| + 1$. Thus, $A$ has a synchronizing set of size at least $\alpha(G) + 1$.
		
		In other direction, let $w$ be a word synchronizing a set of states $S'$ of maximum size in $A$. We can assume that after reading $w$ all the states in $S'$ are mapped to $f$, as all the sets of states that are mapped to any other state have cardinality at most two. Then by construction there are no edges in $G$ between any pair of vertices in $I' = \{v_i \mid s_i \in S'\}$, so $I'$ is an independent set of size $|S'| - 1$ in $G$. Thus the maximum size of a synchronizing set in $A$ is equal to $\alpha(G) + 1$. \end{proof}
	
	Thus we have a gap-preserving reduction from the {\sc Independent Set} problem to the {\sc Max Sync Set} problem with a gap $\Theta(p^{1 - \varepsilon})$ for any $\varepsilon > 0$. It is easy to see that $n = \Theta(p)$ and $A$ is weakly acyclic, which concludes the proof of the theorem. \end{proof}

Next we move to a slightly weaker inapproximability result for binary automata.

\begin{theorem}\label{thm-inapprox-gen}
	The problem {\sc Max Sync Set} for binary $n$-state automata cannot be approximated in polynomial time within a factor of $O(n^{\frac{1}{2} - \varepsilon})$ for any $\varepsilon > 0$ unless P = NP.
\end{theorem}
\begin{proof} Again, we construct a gap-preserving reduction from the {\sc Independent Set} problem extending the proof of Theorem \ref{thm-inapprox-alph}. Given a graph $G = (V, E), V = \{v_1, v_2, \ldots, v_p\}$, we construct an automaton $A =(Q, \Sigma, \delta)$ in the following way. Let $\Sigma = \{0, 1\}$. First we construct the main gadget $A_{main}$ having a synchronizing set of states of size $\alpha(G)$. For each vertex $v_i \in V, 1 \le i \le p$, we construct a set of new states $L_i = V_i \cup U_i$ in~$Q$, where $V_i = \{v^{(i)}_j : 1 \le j \le p\}, U_i = \{u^{(i)}_j : 1 \le j \le p\}$. We call $L_i$ the $i$th {\em layer} of $A_{main}$. We also add a state $f$ to $Q$. For each $i$, $1 \le i \le p$, the transition function $\delta$ imitates choosing taking the vertices $v_1, v_2, \ldots, v_p$ into an independent set one by one and is defined as:
	
	\[ \delta(v^{(i)}_j, 0) = \left\{ 
	\begin{array}{l l}
	u^{(i)}_j & \quad \mbox{if $i = j$,}\\
	v^{(i + 1)}_j & \quad \mbox{otherwise}\\
	\end{array} \right. \]
	
	\[ \delta(v^{(i)}_j, 1) = \left\{ 
	\begin{array}{l l}
	u^{(i)}_j & \quad \mbox{if there is an edge $v_i v_j \in E$,}\\
	v^{(i + 1)}_j & \quad \mbox{otherwise}\\
	\end{array} \right. \]
	
	Here all $v^{(n + 1)}_j, 1 \le j \le p$, coincide with $f$. For each state $u^{(i)}_j$, the transitions for both letters $0$ and $1$ lead to the originating state (i.e. they are self-loops).
	
	We also add an $p$-state cycle $A_{cycle}$ attached to $f$. It is a set of $p$ states $c_1, \ldots, c_p$, mapping $c_i$ to $c_{i + 1}$ and $c_p$ to $c_1$ regardless of the input symbol. Finally, we set $c_1$  to coincide with $f$. Thus we get the automaton $A_1$. Figure \ref{fig:exampleA} presents an example of $A_1$ for a graph with three vertices $v_1, v_2, v_3$ and one edge~$v_2v_3$.
	
	\setlength{\unitlength}{2.2pt}
	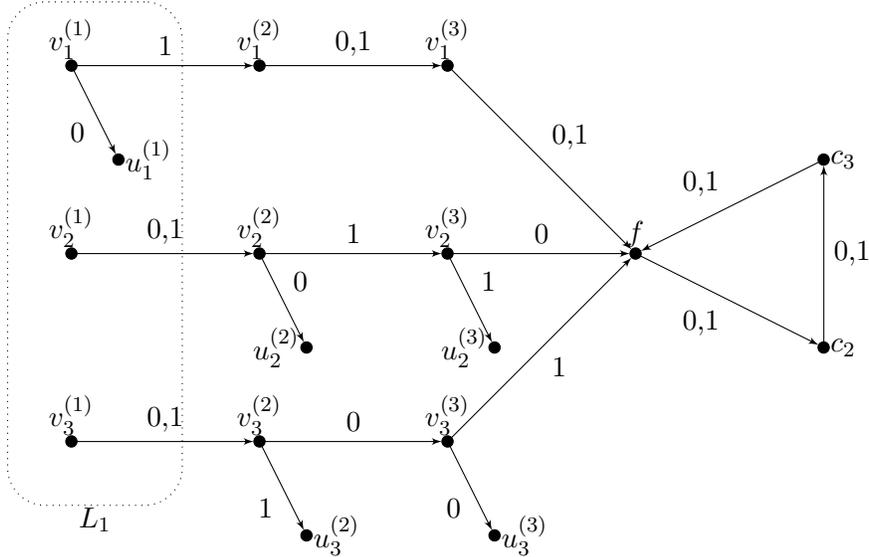
\begin{figure}[hbt]
		\begin{center}
			
			\begin{tikzpicture}[->,>=latex',
			vertex/.style={circle, draw=black, fill, minimum width=1.5mm, inner sep=0pt, outer sep=0pt},
			every label/.style={inner sep=0pt, minimum width=0pt, label distance=0.1mm},
			yscale=-2.5,
			xscale=2.5
			]
			\graph[nodes=vertex, empty nodes, no placement] {
				{
					v11[x=0,y=0,label=above:$v^{(1)}_1$] -> [edge label=1]
					v12[x=1,y=0,label=above:$v^{(2)}_1$] -> [edge  label={$0$,$1$}]
					v13[x=2,y=0,label=above:$v^{(3)}_1$] -> [edge  label={$0$,$1$}]
					f[x=3,y=1,label=above:$f$]
				};
				{
					v11 -> [edge label=0, swap]	u11[x=0.25,y=0.5,label=right:$u^{(1)}_1$]
				};
				{
					v21[x=0,y=1,label=above:$v^{(1)}_2$] -> [edge  label={$0$,$1$}]
					v22[x=1,y=1,label=above:$v^{(2)}_2$] -> [edge label=1]
					v23[x=2,y=1,label=above:$v^{(3)}_2$] -> [edge label=0]
					f
				};
				{
					v22 -> [edge label=0]	u22[x=1.25,y=1.5,label=left:$u^{(2)}_2$]
				};
				{
					v31[x=0,y=2,label=above:$v^{(1)}_3$] -> [edge  label={$0$,$1$}]
					v32[x=1,y=2,label=above:$v^{(2)}_3$] -> [edge label=0]
					v33[x=2,y=2,label=above:$v^{(3)}_3$] -> [edge label=1,swap]
					f
				};
				{
					v32 -> [edge label=1,swap]	u32[x=1.25,y=2.5,label=right:$u^{(2)}_3$]
				};
				{
					v23 -> [edge label=1]	u23[x=2.25,y=1.5,label=left:$u^{(3)}_2$]
				};
				{
					v33 -> [edge label=0, swap]	u33[x=2.25,y=2.5,label=right:$u^{(3)}_3$]
				};
				{
					f -> [edge label={$0$,$1$},swap]
					c1[x=4,y=1.5,label=right:$c_2$] -> [edge label={$0$,$1$},swap]
					c2[x=4,y=0.5,label=right:$c_3$] -> [edge label={$0$,$1$},swap]
					f
				};
			};
			\node[rectangle,dotted,draw,fit=(v11)(u11)(v21)(v31),rounded corners=5mm,inner sep=22pt,label=below:$L_1$] {};
			
			\end{tikzpicture}\textsl{}
			\caption{An example of $A_1$. Unachievable states and self-loops are omitted.} \label{fig:exampleA}
		\end{center}
	\end{figure}
	The main property of $A_1$ is claimed by the following lemma.
	
	\begin{lemma}\label{lemma-main}
		The size of a maximum synchronizing set of states from the first layer in $A_1$ equals~$\alpha(G)$.
	\end{lemma}
	\begin{proof} Let $I$ be a maximum independent set in $G$. Consider a word $w$ of length $p$ such that its $i$th letter is equal to $0$ if $v_i \notin I$ and to $1$ if $v_i \in I$. By the construction of $A_1$, this word synchronizes the set $\{v^{(1)}_j \mid v_j \in I\}$. Conversely, a synchronizing set of at least three states from the first layer can be mapped only to some vertex of $A_{cycle}$, and the corresponding set of vertices in $G$ is an independent set. \end{proof}
	
	Some layer in the described construction can contain a synchronizing subset of size larger than the maximum synchronizing subset of the first layer. To avoid that, we modify $A_1$ by repeating each state (with all transitions) of the first layer $p$ times. More formally, we replace each pair of states $v^{(1)}_j$, $u^{(1)}_j$ with $p$ different pairs of states such that in each pair all the transitions repeat the transitions between $v^{(1)}_j$, $u^{(1)}_j$, and all the other states of the automaton. We denote the automaton thus constructed as $A$.
	
	The following lemma claims that the described procedure of constructing $A$ from $G$ is a gap-preserving reduction from the {\sc Independent Set} problem in graphs to the {\sc Max Sync Set} problem in binary automata.
	
	\begin{lemma}\label{lemma-copy}
		If $\alpha(G) > 1$, then the maximum size of a synchronizing set in $A$ is equal to~$p\alpha(G)~+~1$.
	\end{lemma}
	\begin{proof} Note that due to the construction of $A_{cycle}$, each synchronizing set of $A$ is either a subset of a single layer of $A$ together with a state in $A_{cycle}$ or a subset of a set $\{v^{(i)}_j \mid 2 \le i \le \ell\} \cup \{u^{(\ell)}_j\}$ for some $\ell$ and $j$, together with $p$ new states that replaced $v^{(1)}_j$. Consider the first case. If some maximum synchronizing set $S$ contains a state from the $i$th layer of $A$ and $i > 1$, then its size is at most $p + 1$. A maximum synchronizing set containing some states from the first layer of $A$ consists of $p\alpha(G)$ states from this layer (according to Lemma \ref{lemma-main}) and some state of $A_{cycle}$, so this set has size $p\alpha(G) + 1 \ge 2p + 1$. In the second case, the maximum size of a synchronizing set is at most $p + (p - 1) + 1 = 2p < p \alpha(G) + 1$.
	\end{proof}
	
	It is easy to see that the constructed reduction is gap-preserving with a gap $\Theta(p^{1 - \varepsilon}) = \Theta(n^{\frac{1}{2} - \varepsilon})$, where $n$ is the number of states in $A$, as $n = \Theta(p^2)$. Thus the {\sc Max Sync Set} for $n$-state binary automata cannot be approximated in polynomial time within a factor of $O(n^{\frac{1}{2} - \varepsilon})$ for any $\varepsilon > 0$ unless~P~=~NP, which concludes the proof of the theorem. \end{proof}

Theorem \ref{thm-inapprox-gen} can also be proved by using Theorem \ref{thm-inapprox-alph} and a slight modification of the technique used in \cite{Vorel2016} for decreasing the size of the alphabet. However, in this case the resulting automaton is far from being weakly acyclic, while the automaton in the proof of Theorem \ref{thm-inapprox-alph} has only one cycle. The next theorem shows how to modify our technique to prove an inapproximability bound for {\sc Max Sync Set} in binary weakly acyclic automata.

\begin{theorem}\label{thm-inapprox-monot}
	The {\sc Max Sync Set} problem for binary weakly acyclic $n$-state automata cannot be approximated in polynomial time within a factor of $O(n^{\frac{1}{3} - \varepsilon})$ for any $\varepsilon > 0$ unless P = NP.
\end{theorem}
\begin{proof} We modify the construction of the automaton $A_{main}$ from Theorem \ref{thm-inapprox-gen} in the following way. We repeat each state (with all transitions) of the first layer $p^2$ times in the same way as it is done in the proof of Theorem \ref{thm-inapprox-gen}. Thus we get a weakly acyclic automaton $A_{wa}$ with $n = \Theta(p^3)$ states, where $p$ is the number of vertices in the graph $G$. Furthermore, similar to Lemma \ref{lemma-copy}, the size of the maximum synchronizing set of states in $A_{wa}$ is between $p^2\alpha(G)$ and $p^2\alpha(G) + p(p-1) + 1$, because some of the states from the layers other than the first can be also mapped to $f$. Both of the values are of order $\Theta(p^2\alpha(G))$, thus we have an gap-preserving reduction providing the inapproximability within a factor of $O(p^{1 - \varepsilon}) = O(n^{\frac{1}{3} - \varepsilon})$ for any $\varepsilon > 0$, where $n$ is the number of states in $A_{wa}$.\end{proof}

We finish by noting that for two classes of automata the {\sc Max Sync Set} problem is solvable in polynomial time.

\begin{proposition}\label{thm-polytime-unary}
	The problem {\sc Max Sync Set} can be solved in polynomial time for unary automata.
\end{proposition}
\begin{proof} Consider the digraph $G$ induced by states and transitions of an unary automaton $A$. By definition, each vertex of $G$ has outdegree $1$. Thus, the set of the vertices of $G$ can be partitioned into directed cycles and a set of vertices not belonging to any cycle, but lying on a directed path leading to some cycle. Let $n$ be the number of states in $A$. It is easy to see that after performing $n$ transitions, each state of $A$ is mapped into a state in some cycle, and all further transitions will not map any two different states to the same state. Thus, it is enough to perform $n$ transitions and select such state $s$ that the maximum number of states are mapped to $s$.\end{proof}

A more interesting case is covered by the following proposition. An automaton $A = (Q, \Sigma, \delta)$ is called \textit{Eulerian} if there exists $k$ such that for each state $q \in Q$ there are exactly $k$ pairs $(q', a)$, $q' \in Q$, $a \in \Sigma$, such that $\delta(q', a) = q$.

\begin{proposition}\label{thm-polytime-eulerian}
	The problem {\sc Max Sync Set} can be solved in polynomial time for Eulerian automata.
\end{proposition}
\begin{proof}
	According to Theorem 2.1 in \cite{Friedman1990} (see also \cite{Kari2003} for the discussion of the Eulerian case), each word of minimum rank with respect to an Eulerian automaton synchronizes the sets $S_1, S_2, \ldots, S_r$ forming a partition of the states of the automaton into inclusion-maximal synchronizing sets. Moreover, according to this theorem all inclusion-maximal synchronizing sets in an Eulerian automaton are of the same size, thus each inclusion-maximal synchronizing set has maximum cardinality. A word of minimum rank with respect to an automaton can be found in polynomial time \cite{Rystsov1992}, which concludes the proof.
\end{proof}

\section{Computing the Rank of a Subset of States}\label{sec-rank}

Assume that we know that the current state of the automaton $A$ belongs to a subset $S$ of its states. Even if it is not possible to synchronize $S$, it can be reasonable to minimize the size of the set of possible states of $A$, reducing the uncertainty of the current state as much as possible. One way to do it is to map $S$ to a set $S'$ of smaller size by applying some word to $A$. Recall that the size of the smallest such set $S'$ is called the rank of $S$. Consider the following problem of finding the rank of a subset of states in a given automaton.

\begin{tabular}{||p{30em}}
	~{\sc Set Rank} \\
	~{\em Input}: An automaton $A$ and a set $S$ of its states;\\
	~{\em Output}: The rank of $S$ in $A$. 
\end{tabular}

The rank of an automaton, that is, the rank of the set of its states, can be computed in polynomial time \cite{Rystsov1992}. However, since the automaton in the proof of PSPACE-completeness of {\sc Sync Set} in \cite{Rystsov1983} has rank $2$ (and thus each subset of states in this automaton has rank either $1$ or $2$), it follows immediately that there is no polynomial $c$-approximation algorithm for the {\sc Set Rank} problem for any $c < 2$ unless P = PSPACE. It also follows that checking whether the rank of a subset of states equals the rank of the whole automaton is PSPACE-complete. For monotonic weakly acyclic automata, this problem is hard to approximate within a factor of $\frac{9}{8} - \epsilon$ for any $\epsilon > 0$ \cite{Ryzhikov2017Monotonic}. For general weakly acyclic automata it is possible to get much stronger bounds, as it is shown by the results of this section.

We shall need the {\sc Chromatic Number} problem. A {\em proper coloring} of a graph $G = (V, E)$ is a coloring of the set $V$ in such a way that no two adjacent vertices have the same color. The chromatic number of $G$, denoted $\chi(G)$, is the minimum number of colors in a proper coloring of $G$. Recall that a set of vertices in a graph is called {\em independent} if no two vertices in this set are adjacent. A proper coloring of a graph can be also considered as a partition of the set of its vertices into independent sets.

\begin{tabular}{||p{30em}}
	~{\sc Chromatic Number}\\
	~{\em Input}: A graph $G$;\\
	~{\em Output}: The chromatic number of $G$.
\end{tabular}

This problem cannot be approximated within a factor of $O(p^{1 - \epsilon})$ for any~$\epsilon~>~0$ unless P = NP, where $p$ is the number of vertices in $G$ \cite{Zuckerman2006}.

\begin{theorem}\label{thm-rank}
	The {\sc Set Rank} problem for $n$-state weakly acyclic automata with alphabet of size $O(\sqrt{n})$ cannot be approximated within a factor of $O(n^{\frac{1}{2} - \epsilon})$ for any $\epsilon > 0$ unless P = NP.
\end{theorem}
\begin{proof}
	We shall prove this theorem by constructing a gap-preserving reduction from the {\sc Chromatic Number} problem, extending the technique in the proof of Theorem \ref{thm-inapprox-alph}. Given a graph $G = (V, E)$, $V = \{v_1, v_2, \ldots, v_p\}$, we construct an automaton $A = (Q, \Sigma, \delta)$ as follows. The alphabet $\Sigma$ consists of letters $\tilde{v}_1, \ldots, \tilde{v}_p$ corresponding to the vertices of $G$, together with a {\em switching} letter $\nu$. We use $p$ identical {\em synchronizing} gadgets $T^{(k)}, 1 \le k \le p$, such that each gadget synchronizes a subset of states corresponding to an independent set in $G$. Gadget $T^{(k)}$ consists of a set $\{s^{(k)}_i, t^{(k)}_i \mid 1 \le i \le p\} \cup \{f^{(k)}\}$ of states.
	
	The transition function $\delta$ is defined as follows. For each gadget $T^{(k)}$, for each $1 \le i \le p$, the state $s^{(k)}_i$ is mapped to $f^{(k)}$ by the letter $\tilde{v}_i$. For each $v_iv_j \in E$ the state $s^{(k)}_i$ is mapped to $t^{(k)}_i$ by the letter $\tilde{v}_j$, and the state $s^{(k)}_j$ is mapped to $t^{(k)}_j$ by the letter $\tilde{v}_i$. All yet undefined transitions corresponding to the letters $\tilde{v}_1, \ldots, \tilde{v}_p$ map a state to itself.
	
	It remains to define the transitions corresponding to $\nu$. For each $1 \le k \le p - 1$, $\nu$ maps $t^{(k)}_i$ and $s^{(k)}_i$ to $s^{(k + 1)}_i$, and $f^{(k)}$ to itself. Finally, $\nu$ acts on all states in $T^{(p)}$ as a self-loop.
	
	Define $S = \{s^{(1)}_i \mid 1 \le i \le p\}$. We shall prove that the rank of $S$ is equal to the chromatic number of $G$. Consider a proper coloring of $G$ with the minimum number of colors and let $I_1 \cup \ldots \cup I_{\chi(G)}$ be the partition of $G$ into independent sets defined by this coloring. For each $I_j$, consider a word $w_j$ obtained by concatenating the letters corresponding to the vertices in $I_j$ in some order. Consider now the word $w_1 \nu w_2 \nu \ldots \nu w_{\chi(G)}$. This word maps the set $S$ to the set $\{f^{(i)} \mid 1 \le i \le \chi(G)\}$, which proves that the rank of $S$ is at most $\chi(G)$.
	
	In the other direction, note that after each reading of $\nu$ all states except $f^{(k)}, 1 \le k \le p - 1$, are mapped to the next synchronizing gadget (except the last gadget $T^{(p)}$ which is mapped to itself). By definition of $\delta$, only a subset of states corresponding to an independent set of vertices can be mapped to some particular $f^{(k)}$, and the image of $S$ after reading any word is a subset of the states in some gadget together with some of the states $f^{(k)}, 1 \le k \le p$. Hence, the rank of $S$ is at least $\chi(G)$.
	
	Thus we have a gap-preserving reduction from the {\sc Chromatic Number} problem to the {\sc Set Rank} problem with gap $\Theta(p^{1 - \varepsilon})$ for any $\varepsilon > 0$. It is easy to see that $n = \Theta(p^2)$, $A$ is weakly acyclic and its alphabet has size $O(\sqrt{n})$, which finishes the proof of the theorem.
	 \end{proof}

Using the classical technique of reducing the alphabet size (see \cite{Vorel2016}), $O(n^{\frac{1}{3} - \epsilon})$ inapproximability can be proved for binary automata. To prove the same bound for binary weakly acyclic automata, we have to refine the technique of the proof of the previous theorem.

\begin{theorem}\label{thm-binary-rank}
	The {\sc Set Rank} problem for $n$-state binary weakly acyclic automata cannot be approximated within a factor of $O(n^{\frac{1}{3} - \epsilon})$ for any $\epsilon > 0$ unless P = NP.
\end{theorem}

\begin{proof} To prove this theorem we construct a gap-preserving reduction from the {\sc Chromatic Number} problem, extending the proof of the previous theorem.
	
	Given a graph $G = (V, E), V = \{v_1, v_2, \ldots, v_p\}$, we construct an automaton $A =(Q, \{0, 1\}, \delta)$. In our reduction we use two kinds of gadgets: $p$ {\em synchronizing gadgets} $T^{(k)}$, $1 \le k \le p$, and $p$ {\em waiting gadgets} $R^{(k)}$, $1 \le k \le p$. Gadget $T^{(k)}$ consists of a set $\{v_{i,j}^{(k)} \mid 1 \le i,j \le p\}$  of states, together with a state $f^{(k)}$, and $R^{(k)}$, $1 \le k \le p$, consists of the set $\{u_{i,j}^{(k)} \mid 1 \le i, j \le p\}$.
	
	For each $i, j, k$, $1 \le i, j, k \le p$, the transition function $\delta$ is defined as:
	
	\[ \delta(v^{(k)}_{i, j}, 0) = \left\{ 
	\begin{array}{l l}
	u^{(k)}_{i,j} & \quad \mbox{if $i = j$,}\\
	v^{(k)}_{i + 1,j} & \quad \mbox{otherwise}\\
	\end{array} \right. \]
	
	\[ \delta(v^{(k)}_{i,j}, 1) = \left\{ 
	\begin{array}{l l}
	u^{(k)}_{i,j} & \quad \mbox{if there is an edge $v_i v_j \in E$,}\\
	v^{(k)}_{i + 1,j} & \quad \mbox{otherwise}\\
	\end{array} \right. \]
	
	Here all $v^{(k)}_{p + 1,j}, 1 \le j \le p$, coincide with $f^{(k)}$. We set $\delta(u^{(k)}_{i,j}, x) = u^{(k)}_{i + 1,j}$ for $x \in \{0, 1\}$, $1 \le i, k \le p - 1$, $1 \le j \le p$, and $\delta(u^{(k)}_{p, j}, x) = v^{(k + 1)}_{1,j}$ for $1 \le j \le p$, $1 \le k \le p - 1$, $x \in \{0, 1\}$. The states $u^{(p)}_{i,j}$ are sink states: both letters $0$ and $1$ act on them as self-loops. Finally, we set $S = \{v^{(1)}_{1, j} \mid 1 \le j \le p\}$. Figure \ref{fig-Rank} gives an idea of the described construction.
	
		\setlength{\unitlength}{2.2pt}
		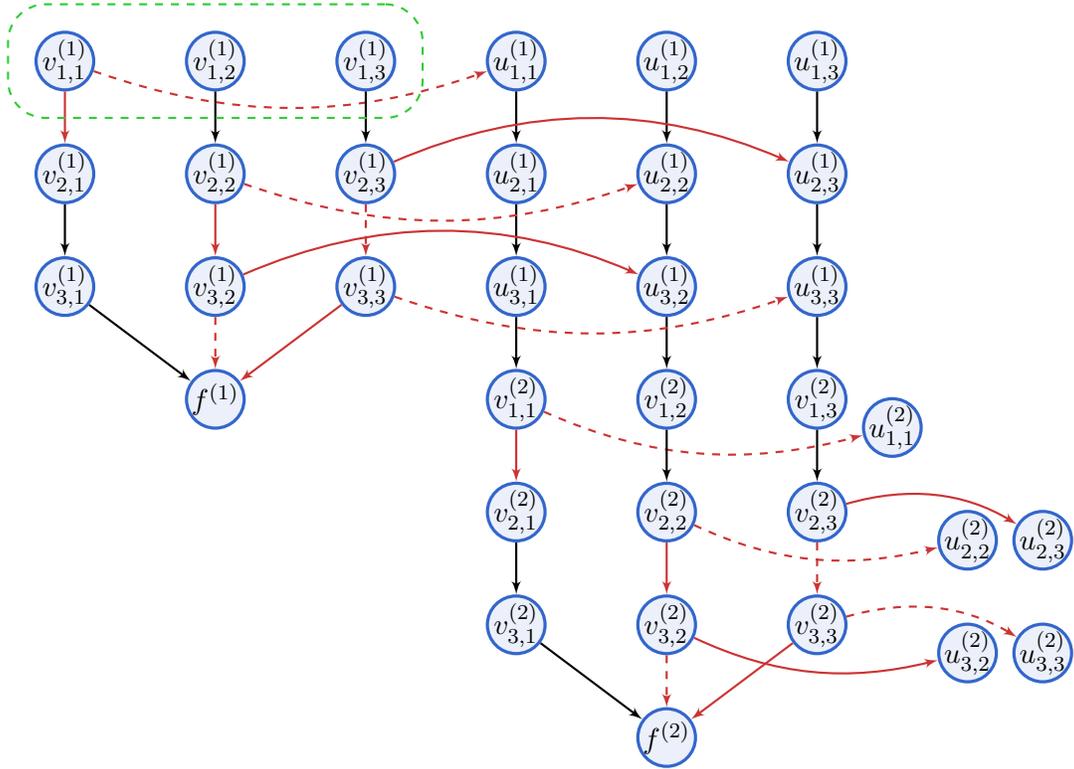
\begin{figure}[hbt]
			\begin{center}
				\begin{tikzpicture}[->,>=latex',
				vertex/.style={circle, draw=black, fill, minimum width=1.5mm, inner sep=0pt, outer sep=0pt},
				every label/.style={inner sep=0pt, minimum width=0pt, label distance=0.1mm},
				yscale=-1.5,
				xscale=2,
				thick
				]
				\graph[nodes=arn_n, empty nodes, no placement] {
					{
						v113[x=3,y=3,label=center:$v^{(2)}_{1, 1}$] -> [edge, bred]
						v123[x=3,y=4,label=center:$v^{(2)}_{2, 1}$] -> [edge]
						v133[x=3,y=5,label=center:$v^{(2)}_{3, 1}$] -> [edge]
						f3[x=4,y=6,label=center:$f^{(2)}$]
					};
					{
						v113 -> [edge, swap, bend left=25, pos=0.25, dashed, bred]
						u113[x=5.5,y=3.25,label=center:$u^{(2)}_{1, 1}$]
					};
					{
						v213[x=4,y=3,label=center:$v^{(2)}_{1, 2}$] -> [edge]
						v223[x=4,y=4,label=center:$v^{(2)}_{2, 2}$] -> [edge, bred]
						v233[x=4,y=5,label=center:$v^{(2)}_{3, 2}$] -> [edge, bred, dashed]
						f3
					};
					{
						v223 -> [edge, bend left=25, pos=0.25, bred, dashed, swap]
						u223[x=6,y=4.25,label=center:$u^{(2)}_{2, 2}$]
					};
					{
						v313[x=5,y=3,label=center:$v^{(2)}_{1, 3}$] -> [edge]
						v323[x=5,y=4,label=center:$v^{(2)}_{2, 3}$] -> [edge, bred, dashed]
						v333[x=5,y=5,label=center:$v^{(2)}_{3, 3}$] -> [edge,swap, bred]
						f3
					};
					{
						v323 -> [edge, bend right, bred]
						u323[x=6.5,y=4.25,label=center:$u^{(2)}_{2, 3}$]
					};
					{
						v233 -> [edge, bend left=25, bred, swap]
						u233[x=6,y=5.25,label=center:$u^{(2)}_{3, 2}$]
					};
					{
						v333 -> [edge, swap, bend right, bred, dashed]
						u333[x=6.5,y=5.25,label=center:$u^{(2)}_{3, 3}$]
					};

					{
						v112[x=3,y=0,label=center:$u^{(1)}_{1, 1}$] -> [edge]
						v122[x=3,y=1,label=center:$u^{(1)}_{2, 1}$] -> [edge]
						v132[x=3,y=2,label=center:$u^{(1)}_{3, 1}$] -> [edge]
						v113
					};
					{
						v212[x=4,y=0,label=center:$u^{(1)}_{1, 2}$] -> [edge]
						v222[x=4,y=1,label=center:$u^{(1)}_{2, 2}$] -> [edge]
						v232[x=4,y=2,label=center:$u^{(1)}_{3, 2}$] -> [edge]
						v213
					};
					{
						v312[x=5,y=0,label=center:$u^{(1)}_{1, 3}$] -> [edge]
						v322[x=5,y=1,label=center:$u^{(1)}_{2, 3}$] -> [edge]
						v332[x=5,y=2,label=center:$u^{(1)}_{3, 3}$] -> [edge]
						v313
					};
					
					{
						v11[x=0,y=0,label=center:$v^{(1)}_{1, 1}$] -> [edge, bred]
						v12[x=0,y=1,label=center:$v^{(1)}_{2, 1}$] -> [edge]
						v13[x=0,y=2,label=center:$v^{(1)}_{3, 1}$] -> [edge]
						f[x=1,y=3,label=center:$f^{(1)}$]
					};
					{
						v11 -> [edge, swap, bend left=25, bred, dashed]	v112
					};
					{
						v21[x=1,y=0,label=center:$v^{(1)}_{1, 2}$] -> [edge]
						v22[x=1,y=1,label=center:$v^{(1)}_{2, 2}$] -> [edge, bred]
						v23[x=1,y=2,label=center:$v^{(1)}_{3, 2}$] -> [edge, bred, dashed]
						f
					};
					{
						v22 -> [edge, bend left=25, bred, dashed] v222
					};
					{
						v31[x=2,y=0,label=center:$v^{(1)}_{1, 3}$] -> [edge]
						v32[x=2,y=1,label=center:$v^{(1)}_{2, 3}$] -> [edge, bred, dashed]
						v33[x=2,y=2,label=center:$v^{(1)}_{3, 3}$] -> [edge,swap, bred]
						f
					};
					{
						v32 -> [edge,swap, bend right, bred]	v322
					};
					{
						v23 -> [edge, bend right, bred]	v232
					};
					{
						v33 -> [edge, swap, bend left=25, dashed, bred]	v332
					};
				};
				\node[rectangle,dashed,draw,fit=(v11)(v21)(v31),
				rounded corners=5mm,inner sep=10pt, bgreen] {};
				
				\end{tikzpicture}\textsl{}
				\caption{A part of the construction in the reduction for {\sc Set Rank}. Red dashed arrows represent transitions for letter $0$, red solid arrows -- for letter $1$, black arrows -- for both letters. Self-loops are omitted.} \label{fig-Rank}
			\end{center}
		\end{figure}
	
	The idea of the presented construction is essentially a combination of the ideas in the proofs of Theorems \ref{thm-inapprox-gen} and \ref{thm-rank}, so we provide only a sketch of the proof. A synchronizing gadget $T^{(k)}$ synchronizes a set $S^{(k)} \subseteq S$ of states corresponding to some independent set in $G$. All the states corresponding to the vertices adjacent to vertices corresponding to $S^{(k)}$ are mapped to the corresponding waiting gadget $R^{(k)}$, and get to the next synchronizing gadget $T^{(k + 1)}$ only after the states of $S^{(k)}$ are synchronized (and thus mapped to $f^{(k)}$). Hence, the minimum size of a partition of $V$ into independent sets is equal to the rank of $S$.  The number of states in $A$ is $O(p^3)$. Thus, we get $\Theta(n^{\frac{1}{3} - \epsilon})$ inapproximability.
	 \end{proof}



\section{Subset Synchronization}\label{sec-subset}

In this section, we obtain complexity results for several problems related to subset synchronization in weakly acyclic automata. We adapt Eppstein's construction from \cite{Eppstein1990}, which is a powerful and flexible tool for such proofs. We shall need the following NP-complete {\sc SAT} problem \cite{Sipser2006}.

\begin{tabular}{||p{32em}}
	~{\sc SAT} \\
	~{\em Input}: A set $X$ of $n$ boolean variables and a set $C$ of $m$ clauses;\\
	~{\em Output}: Yes if there exists an assignment of values to the variables in $X$ such that all clauses in $C$ are satisfied, No otherwise.
\end{tabular}

\begin{theorem}\label{thm-ac-sync-set}
	The {\sc Sync Set} problem in binary weakly acyclic automata is NP-complete.
\end{theorem}
\begin{proof}
	Because of the polynomial upper bound on the length of a shortest word synchronizing a subset of states proved in Proposition \ref{thm-ac-subs-short-word}, we can use such word as a certificate. Thus, the problem is in NP.
	
	We reduce the {\sc SAT} problem. Given $X$ and $C$, we construct an automaton $A =(Q, \{0, 1\}, \delta)$. For each clause $c_j$, we construct $n + 1$ states $y^{(j)}_i, 1 \le i \le n + 1$, in $Q$. We introduce also a state $f \in Q$. The transitions from $y^{(j)}_i$ correspond to the occurrence of $x_i$ in $c_j$ in the following way: for $1 \le i \le n$, $1 \le j \le m$, $\delta(y^{(j)}_i, a) = f$ if the assignment $x_i = a$, $a \in \{0, 1\}$, satisfies $c_j$, and $\delta(y^{(j)}_i, a) = y^{(j)}_{i + 1}$ otherwise. The transition function $\delta$ also maps $y^{(j)}_{n + 1}$ to itself for all $1 \le j \le m$ and both letters $0$ and $1$.
	
	Let $S = \{y^{(j)}_1 \mid 1 \le j \le m\}$. The word $w = a_1 a_2 \ldots a_n$ synchronizes $S$ if $a_i$ is the value of $x_i$ in an assignment satisfying $C$, and vice versa. Thus, the set is synchronizing if and only if all clauses in $C$ can be satisfied by some assignment of binary values to the variables in $X$. 
	 \end{proof}

By identifying the states $y^{(j)}_{n + 1}$ for $1 \le j \le m$ and adding $f$ to $S$ it is also possible to prove that the problem of checking whether the rank of a subset of states equals the rank of an automaton is coNP-complete for binary weakly acyclic automata (cf. the remarks in the beginning of Section \ref{sec-rank}).

The proof of Theorem \ref{thm-ac-sync-set} can be used to prove the hardness of a special case of the following problem, which is PSPACE-complete in general \cite{Kozen1977} and NP-complete for weakly acyclic monotonic automata over a three-letter alphabet \cite{Ryzhikov2017Monotonic}.

\begin{tabular}{||p{32em}}
	~{\sc Finite Automata Intersection} \\
	~{\em Input}: Automata $A_1, \ldots, A_k$ (with initial and accepting states); \\
	~{\em Output}: Yes if there is a word which is accepted by all automata, No otherwise.
\end{tabular}

\begin{proposition}
	{\sc Finite Automata Intersection} is NP-complete when all automata in the input are binary weakly acyclic.
\end{proposition}
\begin{proof} Observe first that if there exists a word which is accepted by all automata, then a shortest such word $w$ has length at most linear in the total number of states in all automata. Indeed, for each automaton consider a topological sort of the set of its states. Each letter of $w$ maps at least one state in some automaton to some other state, which has larger index in the topological sort of the set of states of this automaton. Thus, the considered problem is in NP.
	
	For the hardness proof, we use the same construction as in Theorem \ref{thm-ac-sync-set}. Provided $X$ and $C$, define $A$ in the same way as in Theorem \ref{thm-ac-sync-set}. Define $A_j = (Q_j, \{0, 1\}, \delta_j)$ as follows. Take $Q_j = \{y^{(j)}_i, 1 \le i \le n + 1 \} \cup \{f\}$ and $\delta_j$ to be the restriction of $\delta$ to the set $Q_j$. Set $y^{(j)}_1$ to be the input state and $f$ to be the only accepting state of $A_j$. Then there exists a word accepted by automata $A_1, \ldots, A_m$ if and only if all clauses in $C$ are satisfiable by some assignment.
	 \end{proof}



To obtain the next results, we shall need a modified construction of the automaton from the proof of Theorem \ref{thm-ac-sync-set}, as well as some new definitions. A {\em partial automaton} is a triple $(Q, \Sigma, \delta)$, where $Q$ and $\Sigma$ are the same as in the definition of a finite deterministic automaton, and $\delta$ is a partial transition function (i.e., a transition function which may be undefined for some argument values). Given an instance of the {\sc SAT} problem, construct a partial automaton $A_{base} = (Q, \{0, 1\}, \delta)$ as follows. We introduce a state $f \in Q$. For each clause $c_j$, we construct $n + 1$ states $y^{(j)}_i, 1 \le i \le n + 1$, in $Q$. For each $c_j$, construct also states $z^{(j)}_i$ for $h_i + 1 \le i \le n + 1$, where $h_i$ is the smallest index of a variable occurring in $c_j$. The transitions from $y^{(j)}_i$ correspond to the occurrence of $x_i$ in $c_j$ in the following way: for $1 \le i \le n$, $\delta(y^{(j)}_i, a) = z^{(j)}_{i + 1}$ if the assignment $x_i = a$, $a \in \{0, 1\}$, satisfies $c_j$, and $\delta(y^{(j)}_i, a) = y^{(j)}_{i + 1}$ otherwise. For $x \in \{0, 1\}$, we set $\delta(z^{(j)}_i, a) = z^{(j)}_{i + 1}$ for $h_i + 1 \le i \le n$, $1 \le j \le m$, $a \in \{0, 1\}$. The transition function $\delta$ also maps $z^{(j)}_{n + 1}$, $1 \le j \le m$, and $f$ to $f$ for both letters $0$~and~$1$.

A word $w$ is said to {\em carefully synchronize} a partial automaton $A$ if it maps all its states to the same state $q$, and each mapping corresponding to a prefix of $w$ is defined for each state. The automaton $A$ is then called {\em carefully synchronizing}. We use $A_{base}$ to prove the hardness of the following problem.

\begin{tabular}{||p{30em}}
	~{\sc Careful Synchronization} \\
	~{\em Input}: A partial automaton $A$; \\
	~{\em Output}: Yes if $A$ is carefully synchronizing, No otherwise.
\end{tabular}

For binary automata, {\sc Careful Synchronization} is PSPACE-complete \cite{Martyugin2010}. For monotonic automata over a four-letter alphabet it is NP-hard. We call a partial automaton {\em aperiodic} if for any word $w \in \Sigma^*$ and any state $q \in Q$ there exists $k$ such that either $\delta(q, w^k)$ is undefined, or $\delta(q, w^k) = \delta(q, w^{k + 1})$.

\begin{theorem}\label{thm-careful}
	{\sc Careful Synchronization} is NP-hard for aperiodic partial automata over a three-letter alphabet.
\end{theorem}
\begin{proof}
	We reduce the {\sc SAT} problem. Given $X$ and $C$, we first construct $A_{base}$. Then we add an additional letter $r$ to the alphabet of $A_{base}$ and introduce $m$ new states $s^{(m)}$. For $1 \le i \le n$, $1 \le j \le m$, we define $\delta(s^{(j)}, r) = y^{(j)}_1$, $\delta(y^{(j)}_i, r) = y^{(j)}_1$, $\delta(z^{(j)}_i, r) = y^{(j)}_1$, $\delta(f, r) = f$. All other transitions are left undefined. Let us call the constructed automaton~$A$.
	
	The automaton $A$ is carefully synchronizing if and only if all clauses in $C$ can be satisfied by some assignment of binary values to the variables in $X$. Moreover, the word $w = r w_1 w_2 \ldots w_n 0$, is carefully synchronizing if $w_i$ is the value of $x_i$ in such an assignment.
	
	Indeed, note that the first letter of $w$ is necessarily $r$, as it is the only letter defined for all the states. Moreover, each word starting with $r$ maps $Q$ to a subset of $\{y^{(j)}_i, z^{(j)}_i \mid 1 \le j \le m + 1\} \cup \{f\}$. The only way for a word to map all states to $f$ is to map them first to the set $\{z^{(j)}_{n + 1} \mid 1 \le j \le m\}$, because there are no transitions defined from any $y^{(j)}_{n + 1}$, except the transitions defined by $r$. But this exactly means that there exists an assignment satisfying $C$.
	
	The constructed automaton is aperiodic, because each cycle which is not a self-loop contains exactly one letter $r$.
	 \end{proof}

The complexity of the following problem can be obtained from Theorem \ref{thm-careful}.

\begin{tabular}{||p{32em}}
	~{\sc Positive Matrix}\\
	~{\em Input}: A set $M_1, \ldots, M_k$ of $n \times n$ binary matrices; \\
	~{\em Output}: Yes if there exists a sequence $M_{i_1} \times \ldots \times M_{i_k}$ of multiplications (possibly with repetitions) providing a matrix with all elements equal to $1$, No~otherwise.
\end{tabular}

\begin{corollary}
	{\sc Positive Matrix} is NP-hard for two upper-triangular and two lower-triangular matrices.
\end{corollary}
\begin{proof}
	The proof uses the idea from \cite{Gerencser2016}. Consider three transition matrices corresponding to the letters of the automaton constructed in the proof of Theorem \ref{thm-careful}. Add the matrix corresponding to the letter mapping the state $f$ to all states and undefined for all other states. Any sequence of matrices resulting in a matrix with only positive elements must contain the new matrix, and before that there must be a sequence of matrices corresponding to a word carefully synhronizing the automaton from the proof of Theorem \ref{thm-careful}. Thus we get a reduction from {\sc Careful Synchronization} for aperiodic partial automata over a three-letter alphabet to {\sc Positive Matrix}. It is easy to see that the reduction uses two upper-triangular and two lower-triangular matrices.
\end{proof}

Finally, we show the hardness of the following problem (PSPACE-complete in general \cite{Bondar2016}).

\begin{tabular}{||p{32em}}
	~{\sc Subset Reachability} \\ 
	~{\em Input}: An automaton $A = (Q, \Sigma, \delta)$ and a subset $S$ of its states; \\
	~{\em Output}: Yes if there exists a word $w$ such that $\{\delta(q, w) \mid q \in Q\} = S$, No otherwise.
\end{tabular}

\begin{theorem}
	{\sc Subset Reachability} is NP-complete for weakly acyclic automata.
\end{theorem}
\begin{proof} Consider a topological sort of $Q$.  Let $w$ be a shortest word mapping $Q$ to some reachable set of states. Then each letter of $w$ maps at least one state to a state with a larger index in the topological sort. Thus $w$ has length $O(|Q|^2)$, since the maximum total number of such mappings is $(|Q| - 1) + (|Q| - 2) + \ldots + 1 + 0$. Thus, the considered problem is in NP.
	
	For the NP-hardness proof, we again reduce the {\sc SAT} problem. Given an instance of {\sc SAT}, construct $A_{base}$ first. Next, add a transition $\delta(y^{(j)}_{n + 1}, a) = f$ for $1 \le j \le m$, $a \in \{0, 1\}$, resulting in a deterministic automaton~$A$.
	
	Similar to the proof of Theorem \ref{thm-careful}, $C$ is satisfiable if and only if the set $\{z^{(n + 1)}_j \mid 1 \le j \le m\} \cup \{f\}$ is reachable in $A$. 
	 \end{proof}

\section{Conclusions and Open Problems}

As shown in this paper, weakly acyclic automata serve as an example of a small class of automata where most of the synchronization problems are still hard. More precisely, switching from general automata to weakly acyclic usually results in changing a PSPACE-complete problem to a NP-complete one.

Some problems for weakly acyclic automata are still open. One of them is to study the approximability of the {\sc Shortest Sync Word} problem: there is a drastic gap between known inapproximability results and the $O(n)$-approximation algorithm for general automata. Another natural problem is to study the {\sc Max Sync Set} and {\sc Set Rank} problems complexity in strongly connected automata. The technique used by Vorel for proving PSPACE-completeness of the {\sc Sync Set} problem in strongly connected automata seems to fail here.

\subsubsection*{Acknowledgments}

We would like to thank Peter J. Cameron for introducing us to the notion of synchronizing automata, and Vojt\v{e}ch Vorel, Yury Kartynnik, Vladimir Gusev and Ilia Fridman for very useful discussions. We also thank Mikhail V. Volkov and anonymous reviewers for their great contribution to the improvement of the paper.

{\footnotesize
\bibliography{SyncBib}

\begin{thebibliography}{CLRS09}

\bibitem[AV04]{Ananichev2004}
D.S. Ananichev and M.V. Volkov.
\newblock Synchronizing monotonic automata.
\newblock {\em Theor. Comput. Sci.}, 327(3):225--239, 2004.

\bibitem[Ber14]{Berlinkov2014}
Mikhail~V. Berlinkov.
\newblock On two algorithmic problems about synchronizing automata.
\newblock In Arseny~M. Shur and Mikhail~V. Volkov, editors, {\em DLT 2014.
  LNCS, vol. 8633}, pages 61--67. Springer, Cham, 2014.

\bibitem[BF80]{Brzozowski1980}
J.A. Brzozowski and Faith~E. Fich.
\newblock Languages of {R}-trivial monoids.
\newblock {\em J. Comput. Syst. Sci}, 20(1):32--49, 1980.

\bibitem[BS09]{Bell2009}
Jordan Bell and Brett Stevens.
\newblock A survey of known results and research areas for {N-queens}.
\newblock {\em Discrete Mathematics}, 309(1):1 -- 31, 2009.

\bibitem[BV16]{Bondar2016}
Eugenija~A. Bondar and Mikhail~V. Volkov.
\newblock Completely reachable automata.
\newblock In Cezar C{\^a}mpeanu, Florin Manea, and Jeffrey Shallit, editors,
  {\em DCFS 2016. LNCS, vol. 9777}, pages 1--17. Springer, Cham, 2016.

\bibitem[Car14]{CardosoThesis2014}
Angela Cardoso.
\newblock {\em {The {{\v{C}}ern{\'y}} Conjecture and Other Synchronization
  Problems}}.
\newblock PhD thesis, University of Porto, Portugal, 2014.

\bibitem[CLRS09]{Cormen2009}
Thomas~H. Cormen, Charles~E. Leiserson, Ronald~L. Rivest, and Clifford Stein.
\newblock {\em Introduction to Algorithms}.
\newblock MIT Press, 3rd edition, 2009.

\bibitem[Epp90]{Eppstein1990}
David Eppstein.
\newblock Reset sequences for monotonic automata.
\newblock {\em SIAM J. Comput}, 19(3):500--510, 1990.

\bibitem[Fri90]{Friedman1990}
Joel Friedman.
\newblock On the road coloring problem.
\newblock {\em Proc. Amer. Math. Soc}, 110:1133--1135, 1990.

\bibitem[GGJ16]{Gerencser2016}
Bal{\'{a}}zs Gerencs{\'{e}}r, Vladimir~V. Gusev, and Rapha{\"{e}}l~M. Jungers.
\newblock Primitive sets of nonnegative matrices and synchronizing automata.
\newblock {\em CoRR}, abs/1602.07556, 2016.

\bibitem[GS15]{Gawrychowski2015}
Pawe{\l} Gawrychowski and Damian Straszak.
\newblock Strong inapproximability of the shortest reset word.
\newblock In F.~Giuseppe Italiano, Giovanni Pighizzini, and T.~Donald Sannella,
  editors, {\em MFCS 2015. LNCS, vol. 9234}, pages 243--255. Springer,
  Heidelberg, 2015.

\bibitem[JM12]{Jiraskova2012}
Galina Jir{\'a}skov{\'a} and Tom{\'a}{\v{s}} Masopust.
\newblock On the state and computational complexity of the reverse of acyclic
  minimal {DFAs}.
\newblock In Nelma Moreira and Rog{\'e}rio Reis, editors, {\em CIAA 2012. LNCS,
  vol. 7381}, pages 229--239. Springer, Heidelberg, 2012.

\bibitem[Kar03]{Kari2003}
Jarkko Kari.
\newblock Synchronizing finite automata on eulerian digraphs.
\newblock {\em Theoretical Computer Science}, 295(1):223 -- 232, 2003.
\newblock Mathematical Foundations of Computer Science.

\bibitem[Koz77]{Kozen1977}
Dexter Kozen.
\newblock Lower bounds for natural proof systems.
\newblock In {\em Proceedings of the 18th Annual Symposium on Foundations of
  Computer Science}, pages 254--266. 1977.

\bibitem[Mar10]{Martyugin2010}
P.~V. Martyugin.
\newblock Complexity of problems concerning carefully synchronizing words for
  {PFA} and directing words for {NFA}.
\newblock In Farid Ablayev and Ernst~W. Mayr, editors, {\em CSR 2010. LNCS,
  vol. 6072}, pages 288--302. Springer, Heidelberg, 2010.

\bibitem[Nat86]{Natarajan1986}
B.~K. Natarajan.
\newblock An algorithmic approach to the automated design of parts orienters.
\newblock In {\em Proceedings of the 27th Annual Symposium on Foundations of
  Computer Science}, pages 132--142. 1986.

\bibitem[Pin83]{Pin1983}
Jean-{\'E}ric Pin.
\newblock On two combinatorial problems arising from automata theory.
\newblock {\em {Ann. Discrete Math.}}, 17:535--548, 1983.

\bibitem[RS17]{Ryzhikov2017Monotonic}
Andrew Ryzhikov and Anton Shemyakov.
\newblock Subset synchronization in monotonic automata.
\newblock In Juhani Karhum{\"a}ki and Aleksi Saarela, editors, {\em Proceedings
  of the Fourth Russian Finnish Symposium on Discrete Mathematics, TUCS Lecture
  Notes 26}, pages 154--164. 2017.

\bibitem[Rys83]{Rystsov1983}
Igor~K. Rystsov.
\newblock Polynomial complete problems in automata theory.
\newblock {\em Inform. Process. Lett.}, 16(3):147--151, 1983.

\bibitem[Rys92]{Rystsov1992}
I.~K. Rystsov.
\newblock Rank of a finite automaton.
\newblock {\em Cybern Syst Anal}, 28(3):323--328, 1992.

\bibitem[Rys97]{Rystsov1997}
Igor~K. Rystsov.
\newblock Reset words for commutative and solvable automata.
\newblock {\em Theor. Comput. Sci.}, 172(1):273--279, 1997.

\bibitem[Ryz17]{Ryzhikov2017WA}
Andrew Ryzhikov.
\newblock Synchronization problems in automata without non-trivial cycles.
\newblock In Arnaud Carayol and Cyril Nicaud, editors, {\em CIAA 2017. LNCS,
  vol. 10329}, pages 188--200. Springer, Cham, 2017.

\bibitem[San05]{Sandberg2005}
Sven Sandberg.
\newblock Homing and synchronizing sequences.
\newblock In Manfred Broy, Bengt Jonsson, Joost-Pieter Katoen, Martin Leucker,
  and Alexander Pretschner, editors, {\em Model-Based Testing of Reactive
  Systems: Advanced Lectures. LNCS, vol. 3472}, pages 5--33. Springer,
  Heidelberg, 2005.

\bibitem[Sip12]{Sipser2006}
Michael Sipser.
\newblock {\em Introduction to the Theory of Computation}.
\newblock Cengage Learning, 3rd edition, 2012.

\bibitem[Szy17]{Szykula2017}
Marek Szyku{\l}a.
\newblock Improving the upper bound the length of the shortest reset words.
\newblock {\em CoRR}, abs/1702.05455, 2017.

\bibitem[Tra07]{Trahtman2007}
A.~N. Trahtman.
\newblock {The Cern{\'y} conjecture for aperiodic automata}.
\newblock {\em {Discrete Math. Theor. Comput. Sci.}}, 9(2), 2007.

\bibitem[TY15]{Turker2015}
Uraz~Cengiz T{\"u}rker and H{\"u}sn{\"u} Yenig{\"u}n.
\newblock Complexities of some problems related to synchronizing,
  non-synchronizing and monotonic automata.
\newblock {\em International Journal of Foundations of Computer Science},
  26(01):99--121, 2015.

\bibitem[Vaz01]{Vazirani2001}
Vijay~V. Vazirani.
\newblock {\em Approximation Algorithms}.
\newblock Springer, 2001.

\bibitem[Vol08]{Volkov2008}
Mikhail~V. Volkov.
\newblock Synchronizing automata and the {{\v{C}}ern{\'y}} conjecture.
\newblock In Carlos Mart{\'i}n-Vide, Friedrich Otto, and Henning Fernau,
  editors, {\em LATA 2008. LNCS, vol. 5196}, pages 11--27. Springer,
  Heidelberg, 2008.

\bibitem[Vor14]{Vorel2014}
Vojt\v{e}ch Vorel.
\newblock Subset synchronization of transitive automata.
\newblock In {\em Proceedings 14th International Conference on Automata and
  Formal Languages, {AFL} 2014, Szeged, Hungary, May 27-29, 2014}, pages
  370--381, 2014.

\bibitem[Vor16]{Vorel2016}
Vojtech Vorel.
\newblock Subset synchronization and careful synchronization of binary finite
  automata.
\newblock {\em Int. J. Found. Comput. Sci.}, 27(5):557--578, 2016.

\bibitem[Zuc07]{Zuckerman2006}
David Zuckerman.
\newblock Linear degree extractors and the inapproximability of max clique and
  chromatic number.
\newblock {\em Theory Comput.}, 3(6):103--128, 2007.

\end{thebibliography}
	\bibliographystyle{alpha}}
	
\end{document}